\pdfoutput=1 % arXiv: force pdflatex (PDF figures)
\documentclass[conference]{IEEEtran}
\def\TECHREPORT{1} % arXiv technical report: full-proofs appendix included
\usepackage{amsmath,amssymb}
\usepackage{amsthm}
\usepackage{booktabs}
\usepackage{graphicx}
\usepackage{algorithm}
\usepackage{algpseudocode}
\usepackage{tikz}
\usetikzlibrary{decorations.pathreplacing}
\usepackage{cite}
\usepackage{balance}
\usepackage[hidelinks]{hyperref}

\theoremstyle{plain}
\newtheorem{theorem}{Theorem}
\newtheorem{lemma}{Lemma}
\newtheorem{corollary}{Corollary}
\newtheorem{proposition}{Proposition}

\theoremstyle{definition}
\newtheorem{remark}{Remark}
\newtheorem{example}{Example}

\newcommand{\ResSkewEqual}{20.9}

\newcommand{\ResSkewEnergy}{19.1}

\newcommand{\ResSkewHarmonic}{18.8}

\newcommand{\ResFadingSkewEqual}{51.3}

\newcommand{\ResFadingSkewEnergy}{45.1}

\newcommand{\ResFadingSkewHarmonic}{44.7}

\newcommand{\ResUbindEqual}{17.2}

\newcommand{\ResUbindEnergy}{17.2}

\newcommand{\ResUbindHarmonic}{16.0}

\newcommand{\ResUbindGainE}{6.5}

\newcommand{\ResUbindfadeEqual}{41.6}

\newcommand{\ResUbindfadeEnergy}{41.6}

\newcommand{\ResUbindfadeHarmonic}{39.8}

\newcommand{\ResUbindfadeGainE}{4.3}

{20.9}\newcommand{\ResSkewEnergy}{19.1}
  \newcommand{\ResSkewHarmonic}{18.8}
  
  \newcommand{\ResFadingSkewEqual}{51.3}\newcommand{\ResFadingSkewEnergy}{44.2}
  \newcommand{\ResFadingSkewHarmonic}{44.7}
  
  \newcommand{\ResUbindEqual}{17.2}\newcommand{\ResUbindEnergy}{17.2}
  \newcommand{\ResUbindHarmonic}{16.0}\newcommand{\ResUbindGainE}{6.5}
  \newcommand{\ResUbindfadeEqual}{41.5}\newcommand{\ResUbindfadeEnergy}{41.5}
  \newcommand{\ResUbindfadeHarmonic}{39.4}\newcommand{\ResUbindfadeGainE}{5.0}
}
\newcommand{\ResSkewLB}{15.8}

\newcommand{\ResFadingSkewLB}{35.7}
\newcommand{\ResUbindLB}{15.2}
\newcommand{\ResUbindfadeLB}{35.9}
{34.0}
  \newcommand{\ResSkewLB}{17.8}
  \newcommand{\ResFadingSkewLB}{37.7}
}

\newcommand{\tmin}{T^{\min}}
\newcommand{\dmax}{d_{\max}}
\newcommand{\cbk}{\discretionary{}{}{}}
\newcommand{\TwtPSM}{\texttt{Twt\cbk Power\cbk Save\cbk Manager}}

\ifdefined\TECHREPORT\IEEEoverridecommandlockouts\fi

\title{Age-Optimal Target Wake Time: Provably Good Wake Schedules for
Energy-Constrained Wi-Fi Status Updating}

\author{%
\IEEEauthorblockN{Haoyu Wang}
\IEEEauthorblockA{\textit{Department of Computer Science}\\
\textit{University of Massachusetts Boston}\\
Boston, MA, USA\\
haoyu.wang001@umb.edu}
\and
\IEEEauthorblockN{Bo Sheng}
\IEEEauthorblockA{\textit{Department of Computer Science}\\
\textit{University of Massachusetts Boston}\\
Boston, MA, USA\\
bo.sheng@umb.edu}
\and
\IEEEauthorblockN{Xiaoqian Zhang}
\IEEEauthorblockA{\textit{Department of Computer Science}\\
\textit{University of Nebraska Omaha}\\
Omaha, NE, USA\\
xiaoqianzhang@unomaha.edu}
\ifdefined\TECHREPORT
\thanks{Accepted for publication in the Proceedings of the 28th International
Conference on Modeling, Analysis and Simulation of Wireless and Mobile Systems
(MSWiM '26), Paris, France, October 26--30, 2026.
\copyright~2026 IEEE. Personal use of this material is permitted. Permission
from IEEE must be obtained for all other uses, in any current or future media,
including reprinting/republishing this material for advertising or promotional
purposes, creating new collective works, for resale or redistribution to servers
or lists, or reuse of any copyrighted component of this work in other works.}
\fi
}

\begin{document}
\maketitle

\begin{abstract}
Target Wake Time (TWT), introduced in IEEE 802.11ax, lets an access point
schedule exactly when each station wakes, transmits, and dozes. Existing TWT
schedulers optimize energy or throughput, treating information freshness at
best as a constraint and offering no performance guarantees. We design the
wake schedule itself for freshness: minimize the weighted average Age of
Information (AoI) over stations subject to per-station energy budgets, where
the decision variables are the TWT triples (wake interval, offset, service
period duration). We derive a renewal-exact AoI model for TWT under per-SP
block fading and validate it against packet-level 802.11ax simulation with
${\sim}1\%$ mean error. We show that, unlike preemptive scheduling, non-preemptive TWT
packing can be infeasible at schedule density~1, and identify the granularity
condition under which a small-first best-fit packer provably succeeds. Around
this we build \textsc{Harmonic-Greedy}, a scheduler combining a convex
relaxation, anchor-optimized power-of-two rounding, and a best-of-uniform
safeguard, and prove it is a constant-factor approximation: $4/\ln 2 \approx
5.77$ under a mild granularity assumption and $6/\ln 2 \approx 8.66$
unconditionally. We implement the complete system in ns-3 -- a TWT wake/doze
mechanism integrated with the power-save architecture, plus the scheduler --
and show that it is the only scheduler that stays near a relaxation lower
bound across all regimes: against a strong energy-greedy baseline it ties
when per-station energy floors already pin the periods, and wins by
$4$--$36\%$ exactly where the schedule density is binding and must be
redistributed by AoI weight or channel quality rather than by energy budget
-- the regime our analysis identifies.

\end{abstract}

\begin{IEEEkeywords}
Age of Information; Target Wake Time; IEEE 802.11ax; scheduling;
energy efficiency; ns-3
\end{IEEEkeywords}

\section{Introduction}
\label{sec:intro}

Real-time monitoring applications -- industrial sensing, health wearables,
smart-building telemetry -- send small periodic status updates over Wi-Fi
while running on batteries. Two trends collide in such systems. On one hand,
the relevant performance metric is not throughput or even delay, but the
\emph{Age of Information} (AoI)~\cite{kaul12,yates21}: the time elapsed since
the generation of the freshest delivered update. On the other, the dominant energy cost is simply
being awake: a station that wakes rarely is efficient but stale; one that
wakes often is fresh but dies quickly.

IEEE 802.11ax~\cite{ieee80211ax} introduced Target Wake Time (TWT) precisely
to manage this trade-off mechanically: the AP and a station agree on a
periodic schedule of \emph{service periods} (SPs) -- a wake interval $T$, an offset $\phi$, and an
SP duration $d$ -- and the station dozes outside its SPs. TWT thus exposes,
for the first time in mainstream Wi-Fi, the exact control surface that AoI
theory wants to optimize: \emph{when} each source is allowed to deliver.
Yet existing TWT schedulers optimize energy or throughput, treat freshness at
best as a constraint, and offer no guarantees
(Section~\ref{sec:related}).

This paper designs the TWT schedule itself for freshness, with guarantees.
Our contributions:
\begin{itemize}
\item \textbf{Model} (Sections~\ref{sec:model}--\ref{sec:single}): a
renewal-exact expression for time-average and peak AoI of a TWT station under
per-SP block fading, including the per-wake overhead $c_w$. A structural
lemma shows that batching retransmission attempts within an SP is
\emph{never} beneficial without wake overhead; the per-wake cost $c_w$ is
thus exactly what separates TWT schedule design from abstract sampling
models.
\item \textbf{Hardness structure} (Section~\ref{sec:problem}): a three-station
example with schedule density (the fraction of air time the SPs reserve,
$\sum_i d_i/T_i$) exactly~1 that is infeasible under \emph{every}
placement -- non-preemptive TWT packing fundamentally differs from preemptive
utilization bounds (harmonic rate-monotonic is optimal at density~1 only
because preemption can split service across gap copies).
\item \textbf{Algorithm and guarantees}
(Sections~\ref{sec:algorithm}--\ref{sec:theory}): \textsc{Harmonic-Greedy},
combining a water-filling convex relaxation, power-of-two period rounding
with an anchor chosen by a candidate-argmin rule (evaluate a small, provably
sufficient set of grid anchors and keep the best), small-first best-fit
packing of power-of-two-sized reservations, and a best-of-uniform
safeguard. We prove a
deterministic $4/\ln 2 \approx 5.77$ approximation under a mild granularity
assumption, and $6/\ln 2 \approx 8.66$ unconditionally. The packing analysis
gives a tight granularity threshold $d \le t_0 (1-U)$.
\item \textbf{System} (Section~\ref{sec:implementation}): a complete ns-3
implementation -- a TWT wake/doze mechanism integrated with the 802.11
power-save architecture, with SP-boundary queue gating and TSF-anchored
offsets -- validated against the analytical model with ${\sim}1\%$ mean
error, plus five
protocol findings about deploying TWT (e.g., doze announcements cannot be
acknowledged over a faded channel; pending retransmissions must be gated or
they erase the energy savings).
\item \textbf{Evaluation} (Section~\ref{sec:evaluation}): against both a
naive (equal-interval) and a strong (energy-greedy) baseline, plus a
relaxation lower bound. We separate two effects -- period diversity, which
the energy-greedy baseline already captures, and AoI-weighting, which we
isolate in uniform-budget binding-density regimes where only our scheduler
wins ($4$--$6.5\%$ over both baselines, up to $36\%$ at scale with mixed
modulation rates), exactly where the theory predicts.
\end{itemize}

\section{Related Work}
\label{sec:related}

\textbf{Age of Information.} Since its introduction as a freshness
metric~\cite{kaul12}, AoI has grown into a broad field spanning queueing,
sampling, and estimation~\cite{kosta17,yates21,xoi26}. Foundational results
characterize the freshness-optimal update policy for a single
source~\cite{sun17}, the effect of packet management and queue
discipline~\cite{costa16,bedewy19q,yates19}, and AoI's role in
IoT~\cite{abdelmagid19}. AoI now also serves as a system-level metric in
edge-AI applications -- federated unlearning in mobile edge
computing~\cite{fedfresh26}, digital-twin healthcare~\cite{dtaoi25} --
which presuppose, rather than design, the delivery schedule beneath
them. Variants such as the Age of Incorrect
Information~\cite{maatouk20} refine the metric for estimation. Our renewal
model (Sec.~\ref{sec:single}) specializes this machinery to the TWT
service-period structure.

\textbf{AoI scheduling.} Minimizing weighted AoI across many sources over a
shared channel is the closest body of work. Centralized per-slot policies and
their optimality were established for broadcast networks with reliable and
unreliable links~\cite{kadota18,talak20} and with stochastic
arrivals~\cite{hsu20,kadota21}; the restless-bandit structure yields
Whittle-index policies that are near-optimal and, in cases, asymptotically
optimal~\cite{tripathi24}. These works schedule \emph{which source transmits
in each slot} under a centralized controller; we instead design \emph{periodic
wake agreements} fixed in advance, where a sleeping station is unreachable
between its service periods -- a constraint absent from slot-by-slot models.
A parallel line constrains energy explicitly -- energy-harvesting
sources~\cite{yates15,abdelmagid20}, retransmission/ARQ budgets~\cite{ceran19},
and distributed freshness optimization~\cite{talak18dist} -- but optimizes
transmission timing rather than a negotiated periodic wake schedule with a
hard duty-cycle budget. Closest in spirit are sleep-wake AoI optimization for
CSMA sensors~\cite{bedewy20}, which tunes continuous wake \emph{rates} under
battery constraints but has no negotiated periodic structure, and age-agnostic
cyclic schedulers~\cite{gamgam23}, which optimize deterministic transmission
patterns but model neither energy nor sleeping. A pull-based line measures
freshness only when the monitor asks: the Query Age of
Information~\cite{qaoi22} and Whittle-index policies that optimize age at
query instants~\cite{queryaoi24}. TWT is structurally dual: the delivery
opportunities, not the queries, occur at negotiated instants, and our
objective remains the continuous-time average.

\textbf{TWT and Wi-Fi.} TWT was introduced in IEEE 802.11ax to manage
energy and contention~\cite{ieee80211ax,bellalta16,khorov19}, and is extended
in Wi-Fi 7 (802.11be) -- notably restricted TWT for latency-sensitive
traffic and multi-link operation~\cite{lopezperez19,garciarodriguez21} --
making schedule design increasingly central~\cite{twtsurvey25}.
TASPER~\cite{tasper25} schedules TWT SPs to maximize throughput and energy
efficiency subject to AoI \emph{constraints}, proves NP-hardness, and gives a
heuristic without guarantees; timely-throughput TWT grouping maximizes
deadline-met throughput via drift-plus-penalty~\cite{twttimely23};
experimental IIoT schedulers~\cite{twtdet25}
respect AoI deadlines deterministically; and TWT mechanism implementations
and measurement studies~\cite{wns3twt24a,wns3twt24b,nurchis19} provide the
platform context. Single-device closed-form AoI analyses exist for the
cellular sleep-mode analogue of TWT (DRX)~\cite{drxaoi23} and, most
recently, for the TWT mechanism itself~\cite{ting26} -- performance
analysis of given configurations, not schedule design across a station
\emph{population}, which is what we do here. To our knowledge no prior
work makes AoI the \emph{objective} of TWT schedule design or provides
approximation guarantees.

\textbf{Periodic real-time scheduling.} Our packing analysis connects to
classical real-time theory. Rate-monotonic scheduling is optimal for
preemptive periodic tasks~\cite{liu73}, and harmonic task sets are schedulable
up to utilization~1 -- but only because preemption can split a job across the
gaps left by higher-rate tasks. TWT service periods are \emph{contiguous and
non-preemptive}, which is exactly why our Proposition~\ref{prop:ce} shows
density~1 can be infeasible. The relevant non-preemptive analogue is the
pinwheel problem and distance-constrained scheduling~\cite{holte89,han92,chan92},
whose density thresholds and harmonic constructions inform our
small-first best-fit packer; we adapt them to the AoI-weighted, energy-budgeted
TWT setting.

\section{System Model}
\label{sec:model}

$N$ stations are associated with one AP. Station $i$ holds a TWT agreement
$(T_i, \phi_i, d_i)$, where $T_i$ is its \emph{wake interval} (the time
between consecutive wake-ups), $\phi_i$ its \emph{offset} (the start of its
first SP on the shared clock), and $d_i$ its \emph{SP duration} (the
contiguous awake time per wake): the station wakes at times $\phi_i + k T_i$,
$k = 0, 1, \dots$, stays awake for $d_i$, and dozes otherwise
(Fig.~\ref{fig:mechanism}; Table~\ref{tab:notation} collects notation). SPs
of different stations must not overlap: keeping them disjoint is the
scheduler's responsibility, and it is precisely what designs contention
away. Station $i$ samples its source at the start of each SP
(generate-at-will) and transmits the update uplink within the SP.

\begin{table}[t]
\centering
\caption{Notation.}
\label{tab:notation}
\footnotesize
\begin{tabular}{ll}
\toprule
$(T_i, \phi_i, d_i)$ & TWT triple: wake interval, offset, SP duration \\
$\rho_i$;\ \ $\tmin_i$ & duty-cycle budget; energy floor $d_i/\rho_i$ \\
$p_i$ & per-SP success probability (block fading) \\
$q_i$, $k_i$, $s_i$ & per-attempt success, attempts per SP, attempt airtime \\
$c_{w,i}$ & per-wake overhead (ramp-up, synchronization) \\
$\delta_i$ & mean in-SP delivery delay \\
$w_i$;\ \ $a_i$ & AoI weight; cost coefficient $w_i(1/p_i - 1/2)$ \\
$\bar A_i$, $\bar A_i^{\mathrm{peak}}$ & time-average and peak AoI \\
$U$ & schedule density $\sum_i d_i / T_i$ \\
$\dmax$ & largest SP duration $\max_k d_k$ \\
$t_0$, level $j$ & harmonic grid anchor; grid period $t_0 2^j$ \\
$\rho^\star$;\ \ $\lambda$ & scheduler density target; water-filling multiplier \\
\bottomrule
\end{tabular}
\end{table}

\begin{figure}[t]
\centering
% F1: TWT mechanism timeline (two stations, SPs on a shared grid)
\begin{tikzpicture}[xscale=0.74, yscale=0.62, font=\scriptsize]
  % time axis
  \draw[->] (0,0) -- (10.6,0) node[below left] {time};

  % station 2 (period 2T): row at y=1
  \node[left] at (0,1.1) {STA 2};
  \draw[gray!50] (0,0.75) -- (10.2,0.75);
  \foreach \x in {1.5, 6.3} {
    \fill[blue!25] (\x,0.75) rectangle (\x+0.9,1.45);
    \draw[blue!60!black] (\x,0.75) rectangle (\x+0.9,1.45);
  }
  \node[blue!60!black] at (1.95,1.13) {SP};
  \draw[<->] (2.5,1.1) -- (6.3,1.1) node[midway, above=-1pt] {$T_2 = 2T_1$};

  % station 1 (period T): row at y=2.6
  \node[left] at (0,2.7) {STA 1};
  \draw[gray!50] (0,2.35) -- (10.2,2.35);
  \foreach \x in {0.3, 2.7, 5.1, 7.5, 9.9} {
    \fill[orange!30] (\x,2.35) rectangle (\x+0.6,3.05);
    \draw[orange!70!black] (\x,2.35) rectangle (\x+0.6,3.05);
  }
  \draw[<->] (0.3,3.35) -- (2.7,3.35) node[midway, above=-1pt] {$T_1$};
  \draw[decorate,decoration={brace,mirror,raise=2pt}] (0.3,2.3) -- (0.9,2.3)
    node[midway, below=4pt] {$d_1$};

  % doze annotation (between STA 1 SPs)
  \node[gray] at (8.7,1.9) {doze (radio off)};
  \draw[gray,->] (8.2,2.1) -- (6.9,2.6);

  % offset annotation
  \draw[dashed, gray] (0.3,-0.15) -- (0.3,2.35);
  \draw[dashed, gray] (1.5,-0.15) -- (1.5,0.75);
  \draw[<->] (0.3,-0.35) -- (1.5,-0.35) node[midway, below=-1pt] {$\phi_2 - \phi_1$};
\end{tikzpicture}
\caption{TWT wake schedules: station $i$ wakes for $d_i$ every $T_i$ at
offset $\phi_i$ on a common (TSF-anchored) grid; the scheduler chooses the
triples so that service periods are disjoint.}
\label{fig:mechanism}
\end{figure}
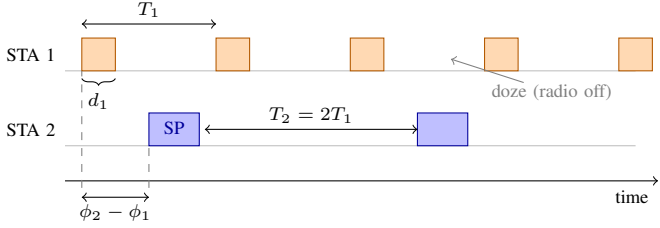

\textbf{Channel.} We model per-SP block fading: the channel coherence time
exceeds the SP duration, so all transmission attempts within one SP share
fate; an SP succeeds with probability $p_i \in (0,1]$, independently across
SPs. This is the regime in which the retransmission policy genuinely
matters for AoI. At the opposite extreme -- per-attempt i.i.d.\ loss --
a few in-SP retries drive the per-SP success probability to~1 and the
channel dimension of the problem disappears; we verify this empirically in
Section~\ref{sec:evaluation}.

\textbf{Energy.} The station's energy cost is its awake time: each wake costs
a fixed overhead $c_{w,i}$ (radio ramp-up, synchronization) plus the SP
airtime. Station $i$ has a duty-cycle budget $\rho_i$:
$d_i / T_i \le \rho_i$, equivalently $T_i \ge \tmin_i := d_i/\rho_i$.

\textbf{Objective.} Let $A_i(t)$ be the age of the freshest update from
station $i$ delivered to the AP. We minimize the weighted time-average AoI
$\sum_i w_i \bar A_i$; Section~\ref{sec:theory} notes that all results carry
to weighted peak AoI.

\textbf{Scope.} We deliberately study the cleanest setting that isolates
schedule design: uplink status updates, generate-at-will sampling, scheduled
non-overlapping SPs (so contention is designed away), and per-SP block
fading. These assumptions are standard in AoI scheduling~\cite{kadota18,
gamgam23} and matched to monitoring IoT, the dominant TWT use case. Three
violations deserve comment. \emph{Stochastic arrivals:} if updates arrive
randomly rather than on demand, the wake schedule itself stays feasible
unchanged, but the sampled update may already be stale at wake time; the
objective acquires an arrival-age component that couples the stations only
through the density constraint, turning the design into a constrained-MDP
problem (Section~\ref{sec:conclusion}). \emph{Contention:} traffic from
non-TWT stations can collide with SPs; such losses enter the model through
$p_i$ as long as they are independent across SPs, and where that cannot be
assumed, 802.11be restricted TWT provides protected SPs. \emph{Intra-SP
channel variation:} block fading is the boundary case in which
retransmission policy matters at all -- under per-attempt i.i.d.\ loss,
in-SP retries drive the SP success probability to 1 and the problem
degenerates gracefully (Section~\ref{sec:single}; verified empirically in
Section~\ref{sec:evaluation}). Downlink AoI (broadcast TWT, AP-side
delivery) is likewise deferred to Section~\ref{sec:conclusion}.

\section{Single-Station Structure}
\label{sec:single}

Successful SPs form a renewal process; the standard renewal-reward AoI
argument~\cite{kaul12,yates21} then gives the following, which is the linchpin
that turns TWT scheduling into a linear objective (Section~\ref{sec:problem}).

\begin{proposition}[TWT AoI]
\label{prop:renewal}
With wake interval $T$, per-SP success probability $p$, and mean in-SP
delivery delay $\delta$,
\[
\bar A = \delta + T\Big(\frac{1}{p} - \frac12\Big), \qquad
\bar A^{\mathrm{peak}} = \delta + \frac{T}{p}.
\]
\end{proposition}

\begin{proof}[Sketch]
Let $G$ be the number of wake intervals between successive successful SPs.
Since each SP succeeds independently with probability $p$, $G$ is geometric on
$\{1,2,\dots\}$, so $\mathbb{E}[G] = 1/p$ and $\mathbb{E}[G^2] = (2-p)/p^2$.
The inter-delivery time is $X = GT$. After a successful delivery the AoI drops
to the in-SP delay (mean $\delta$) and then grows linearly until the next
success, so each renewal cycle contributes area $\delta X + X^2/2$. Dividing
expected area by expected cycle length gives
$\bar A = \delta + \mathbb{E}[X^2]/(2\mathbb{E}[X]) = \delta + T(1/p - 1/2)$,
and the per-cycle peak $\delta + X$ averages to
$\bar A^{\mathrm{peak}} = \delta + \mathbb{E}[X] = \delta + T/p$.
\end{proof}

When the channel is per-attempt i.i.d. with success probability $q$ and the
SP holds $k$ attempts, $p = 1-(1-q)^k$ and $d = k s + c_w$ with $s$ the
per-attempt airtime. Choosing $k$ under the energy-tight period
$T = d/\rho$ reveals a clean dichotomy:

\begin{lemma}[batching needs overhead]
\label{lem:batching}
With $c_w = 0$, the single attempt $k^*=1$ is optimal for every $q$: writing
$f(k) = (k + c_w/s)(1/p(k) - 1/2)$, one computes
$f(2)-f(1) = q/(2(2-q)) > 0$. With $c_w > 0$, $k^* > 1$ for sufficiently
large $c_w/s$ at any fixed $q$.
\end{lemma}

The per-wake overhead is thus load-bearing: it is exactly what distinguishes
TWT schedule design from abstract sampling models, and we keep it in the
model throughout.

We validated Proposition~\ref{prop:renewal} against packet-level 802.11ax
simulation (Section~\ref{sec:evaluation}): measured time-average AoI
matches within $0.02\%$ ($p \approx 1$, two parameter sets) and within
$0.6$--$1.2\%$ (seed-averaged) under per-SP block fading with
$p \in \{0.8, 0.6, 0.4\}$,
with measured duty cycles equal to $d/T$ to four decimals.

\section{The Multi-Station Problem}
\label{sec:problem}

By Proposition~\ref{prop:renewal} the schedule cost is
$J(T) = \sum_i a_i T_i + C$ with $a_i := w_i(1/p_i - 1/2)$ and $C$ a
$T$-independent constant; the design problem is
\begin{align}
\min_{T, \phi}\ & \textstyle\sum_i a_i T_i \nonumber\\
\text{s.t.}\ & \text{SPs pairwise disjoint},\quad T_i \ge \tmin_i.
\label{eq:problem}
\end{align}
Disjointness implies the \emph{density} bound $U = \sum_i d_i/T_i \le 1$,
but the converse fails in a structural way:

\begin{proposition}[density 1 can be unpackable]
\label{prop:ce}
Let $t_0$ be a base period and consider stations $A_1$: $(T{=}t_0,
d{=}t_0/2)$, $A_2$: $(T{=}t_0, d{=}t_0/4)$, $B$: $(T{=}2t_0, d{=}t_0/2)$,
with $U = 1$. No placement of offsets makes the SPs disjoint.
\end{proposition}

\begin{proof}
Both $A_1$ and $A_2$ have period $t_0$, so their SPs trace the same
pattern in every base window $[k t_0, (k{+}1) t_0)$, and the free set $S$
left in each window, of measure $t_0/4$, is likewise identical from
window to window. Now place $B$: its SP is a
contiguous interval of length $t_0/2 \le t_0$, so it meets at most two
consecutive windows, splitting into (at worst) a suffix of one window and
a prefix of the next. Each piece must avoid $A_1$ and $A_2$, i.e.\ lie
inside that window's copy of $S$; and since $S$ repeats identically, the
two pieces occupy \emph{disjoint} parts of the same per-window set. Their
lengths sum to $t_0/2$, forcing $|S| \ge t_0/2 > t_0/4$ -- a contradiction
(Fig.~\ref{fig:ce1}).
\end{proof}

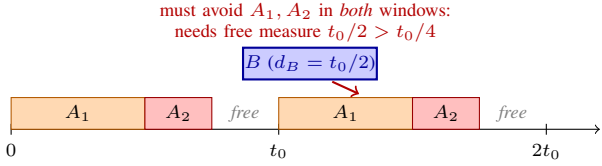
\begin{figure}[t]
\centering
% F3: CE1 -- density 1 that no placement can pack
\begin{tikzpicture}[xscale=1.18, yscale=0.6, font=\scriptsize]
  % two base windows [0,t0), [t0,2t0)
  \draw[->] (0,0) -- (6.6,0);
  \foreach \x/\l in {0/0, 3/{t_0}, 6/{2t_0}} {
    \draw (\x,-0.12) -- (\x,0.12);
    \node[below] at (\x,-0.12) {$\l$};
  }
  % A1 (d = t0/2) and A2 (d = t0/4) repeat in EVERY window
  \foreach \o in {0, 3} {
    \fill[orange!30] (\o,0) rectangle (\o+1.5,0.7);
    \draw[orange!70!black] (\o,0) rectangle (\o+1.5,0.7);
    \fill[red!25] (\o+1.5,0) rectangle (\o+2.25,0.7);
    \draw[red!60!black] (\o+1.5,0) rectangle (\o+2.25,0.7);
  }
  \node at (0.75,0.35) {$A_1$};
  \node at (1.875,0.35) {$A_2$};
  \node at (3.75,0.35) {$A_1$};
  \node at (4.875,0.35) {$A_2$};
  % free gaps
  \foreach \o in {0, 3} {
    \node[gray] at (\o+2.625,0.35) {\emph{free}};
  }
  % B's needed contiguous t0/2 interval -- shown failing
  \fill[blue!20] (2.6,1.1) rectangle (4.1,1.8);
  \draw[blue!60!black, thick] (2.6,1.1) rectangle (4.1,1.8);
  \node[blue!60!black] at (3.35,1.45) {$B$ ($d_B = t_0/2$)};
  \draw[->, thick, red!70!black] (3.6,1.05) -- (3.9,0.78);
  \node[red!70!black, align=center] at (3.3,2.35)
    {must avoid $A_1, A_2$ in \emph{both} windows:\\
     needs free measure $t_0/2 > t_0/4$};
\end{tikzpicture}
\caption{Proposition~\ref{prop:ce}: at density $1$, the per-window free set
has measure $t_0/4$, but $B$'s contiguous SP needs $t_0/2$ of it -- in every
window, under every placement.}
\label{fig:ce1}
\end{figure}

\begin{remark}
Under \emph{preemptive} scheduling the same instance is feasible (harmonic
rate-monotonic is optimal at utilization~1): preemption splits $B$'s service
across the two gap copies. TWT SPs are contiguous by design, so utilization
intuition from real-time scheduling does not transfer. Any packing guarantee
must involve granularity, which our analysis makes precise
(Lemma~\ref{lem:packing}).
\end{remark}

A second structural fact bounds every period from below:

\begin{lemma}[long SPs dominate periods]
\label{lem:longsp}
In any feasible schedule, $T_i \ge \dmax := \max_k d_k$ for all $i$.
\end{lemma}

\begin{proof}
Suppose some station $A$ had $T_A < d_B$ for a station $B$. The wake
starts of $A$ are spaced $T_A$ apart, so at least one falls strictly
inside $B$'s SP of length $d_B$, and the $A$-SP beginning there overlaps
it. Hence every period is at least the longest SP duration.
\end{proof}

The following small instance is carried through
Sections~\ref{sec:algorithm}--\ref{sec:theory} to make each stage of the
design concrete.

\begin{example}[running]
\label{ex:run}
Three stations with equal SPs $d_i = 4$\,ms, uniform duty budgets
$\rho_i = 0.5$ (energy floors $\tmin_i = 8$\,ms), reliable channels
($p_i = 1$, so $a_i = w_i/2$), and weights $w = (8, 2, 1)$, i.e.
$a = (4, 1, \tfrac12)$. At the floors the density would be
$3 \cdot 4/8 = 1.5 > 1$: the air time the stations \emph{want} exceeds what
exists, so the density constraint binds and the periods must be
differentiated -- by weight, not by budget, since the budgets are
identical: exactly the regime our evaluation isolates
(Section~\ref{sec:evaluation}).
\end{example}

\section{The \textsc{Harmonic-Greedy} Scheduler}
\label{sec:algorithm}

\begin{algorithm}[t]
\caption{\textsc{Harmonic-Greedy}}
\label{alg:hg}
\begin{algorithmic}[1]
\Require per station $i$: weight $w_i$, duty budget $\rho_i$, per-attempt
success $q_i$, attempt airtime $s_i$, wake overhead $c_{w,i}$; density
target $\rho^\star$
\Ensure disjoint TWT triples $(T_i, \phi_i, d_i)$ with $d_i/T_i \le \rho_i$
\State \emph{size SPs:} $k_i \gets$ energy-tight optimal attempts
(Lemma~\ref{lem:batching}); $d_i \gets k_i s_i + c_{w,i}$;
$p_i \gets 1 - (1-q_i)^{k_i}$; $a_i \gets w_i(1/p_i - \tfrac12)$;
$\tmin_i \gets d_i/\rho_i$
\State \emph{relax:} $T_i \gets \max(\tmin_i, \dmax, \sqrt{\lambda d_i/a_i})$,
bisecting $\lambda \ge 0$ until $\sum_i d_i/T_i \le \rho^\star$
(Lemma~\ref{lem:relax})
\State \emph{anchor and round:} fold each $T_i$ by powers of two into
$(T_{\min}/2,\, T_{\min}]$, $T_{\min} = \min_i T_i$; among these candidate
anchors pick the $t_0$ minimizing $\sum_i a_i \lceil T_i \rceil_{t_0}$, where
$\lceil \cdot \rceil_{t_0}$ rounds up on the grid $\{t_0 2^j\}$
(Corollary~\ref{cor:anchor}); set $T_i \gets \lceil T_i \rceil_{t_0}$
\State \emph{re-expand:} for stations with $p_i < 1$, grow $k_i$ (hence
$d_i$) into the rounding slack, capped by $\rho_i$ -- extra attempts at
fixed $T_i$ are free for AoI
\State \emph{pack:} in increasing-$T_i$ order, place each $\phi_i$ by
small-first best-fit (certified variant: dyadic reservations; practical
variant: first-fit of the true $d_i$; Section~\ref{sec:theory})
\State \emph{safeguard:} $T_u \gets \max(\max_i \tmin_i, \sum_i d_i)$; if
$\sum_i a_i T_u$ beats the rounded cost, or packing failed, return the
uniform schedule (offsets staggered back-to-back) instead
\end{algorithmic}
\end{algorithm}

Algorithm~\ref{alg:hg} lists the scheduler.
Line~1 sizes each station's SP by the single-station dichotomy of
Section~\ref{sec:single} and precomputes the cost coefficients $a_i$.
Line~2 solves the convex relaxation below; its water-filling solution
shortens the periods of high-weight, poor-channel stations first. Line~3
snaps the ideal periods onto a power-of-two grid -- this is what makes the
later packing tractable -- with the grid anchor chosen by the
candidate-argmin rule analyzed in Corollary~\ref{cor:anchor}. Line~4
reinvests the slack that rounding up creates into extra transmission
attempts. Line~5 places offsets, and line~6 guards against the
octave-quantization pathology discussed at the end of this section. The
relaxation of line~2 replaces disjointness by its density consequence:
\begin{equation}
J^* = \min\Big\{ \sum_i a_i T_i : \sum_i \frac{d_i}{T_i} \le 1,\;
T_i \ge \max(\tmin_i, \dmax) \Big\},
\label{eq:relax}
\end{equation}
Every feasible schedule satisfies the constraints of \eqref{eq:relax} --
disjoint SPs occupy at most the full timeline, and
Lemma~\ref{lem:longsp} supplies the $\dmax$ floor -- so $J^*$ lower-bounds
\eqref{eq:problem}.

\begin{lemma}[water-filling]
\label{lem:relax}
\eqref{eq:relax} is convex; its solution is
$T_i(\lambda) = \max\big(\tmin_i, \dmax, \sqrt{\lambda d_i / a_i}\big)$ with
$\lambda \ge 0$ the density multiplier (found by bisection).
\end{lemma}

\begin{proof}[Sketch]
The objective is linear and each $d_i/T_i$ is convex on $T_i > 0$, so
\eqref{eq:relax} is a convex program and KKT conditions are sufficient. For
fixed $\lambda$ the Lagrangian separates across stations; its $i$-th term
$a_i T_i + \lambda d_i/T_i$ is strictly convex with unconstrained minimizer
$\sqrt{\lambda d_i/a_i}$, so the coordinatewise minimizer over the feasible
ray is that value clipped at the floor $\max(\tmin_i, \dmax)$ -- exactly
$T_i(\lambda)$. The density $\sum_i d_i/T_i(\lambda)$ is continuous and
nonincreasing in $\lambda$, so bisection finds the smallest $\lambda$ at
which it meets the density budget ($\lambda = 0$ if the floors alone
satisfy it), which is complementary slackness. Full proof:
\ifdefined\TECHREPORT appendix\else technical report\fi.
\end{proof}

Two implementation-driven design elements deserve emphasis, both discovered
by simulation. First, the \emph{anchor matters}: rounding to powers of two
of an arbitrary base can double a period. Writing $T_{\min} = \min_i T_i$
for the smallest relaxed period, we instead choose the anchor among the
ideal periods folded into $(T_{\min}/2, T_{\min}]$ -- each candidate makes
one station's rounding lossless -- which is provably no worse than the
average over random anchors (Corollary~\ref{cor:anchor}) and empirically
much better. Second, the
\emph{safeguard matters}: when the relaxation's period spread is sub-octave
(all ideal periods within a factor 2), the harmonic grid cannot express the
differentiation and rounding losses dominate. For example, ideal periods of
$30$ and $50$\,ms (ratio $1.67 < 2$) round to the \emph{same} power-of-two
grid level, so the harmonic schedule cannot give the more important station a
shorter period and merely pays the rounding cost; the safeguard then returns
the uniform schedule instead. We call this the \emph{octave quantization}
effect and quantify it in Section~\ref{sec:evaluation}.

\begin{example}[continued]
On the instance of Example~\ref{ex:run}, water-filling clamps station~1 to
its floor and spreads the remaining density by weight:
$T^* = (8, 13.7, 19.3)$\,ms, density exactly 1, relaxed cost
$J^* = \sum_i a_i T_i^* = 55.3$ (a lower bound). Folding $T^*$ into $(4, 8]$
gives the anchor candidates $\{8, 6.83, 4.83\}$; rounding up on each
candidate grid costs $\{64.0,\, 81.9,\, 67.6\}$, so line~3 picks $t_0 = 8$
and $T = (8, 16, 32)$\,ms -- station 1's rounding is lossless and the cost
is $1.16\times$ the relaxation bound. The uniform alternative
$T_u = \max(8, 3 \cdot 4) = 12$\,ms would cost $66.0 > 64.0$, so the
safeguard keeps the harmonic schedule. Packing in increasing-period order
places $\phi = (0, 4, 12)$\,ms, disjoint over the $32$\,ms hyperperiod.
\end{example}

\section{Approximation Guarantees}
\label{sec:theory}

The practical scheduler packs true SP durations at density $0.9$. A
worst-case certificate for that exact configuration would have to survive
adversarial duration mixes; we instead analyze a deliberately conservative
variant of Algorithm~\ref{alg:hg} and charge the difference to the
constants. The analyzed variant differs in two ways. First, it runs the
relaxation at a density target $\rho^\star$ (fixed to $1/4$ in
Theorem~\ref{thm:approx}), leaving the packing stage guaranteed headroom.
Second, it packs dyadic \emph{reservations}: each station reserves the
smallest grid divisor $t_0/2^m \ge d_i$ and transmits only for $d_i$ inside
its reservation, so reservations consume schedule space but no energy.
Packing is small-first best-fit (SF-BF): process stations by increasing
period; place each into a minimum-size free block of sufficient size,
carving from the left and returning the staircase remainder. The analysis
then proceeds in three steps: a structural invariant on the free blocks
SF-BF maintains (Lemma~\ref{lem:counts}), a sufficient packing threshold
(Lemma~\ref{lem:packing}), and the end-to-end guarantees
(Theorems~\ref{thm:approx} and~\ref{thm:uncond}, sharpened by the anchor
rule of Corollary~\ref{cor:anchor}).

\begin{lemma}[count invariant]
\label{lem:counts}
When SF-BF serves a request of size $q$ from a chosen free block of size
$b$, the staircase blocks it returns have sizes in $[q, b)$ -- sizes at
which, by best-fit minimality of $b$, no free block currently exists.
Consequently new free blocks are only created at sizes whose count is
zero, and at the start of level-$j$ processing every size has at most
$2^j$ free blocks. Both bounds are tight.
\end{lemma}

\begin{lemma}[packing threshold]
\label{lem:packing}
If every reservation satisfies $d'_i \le t_0 (1 - U)$ with
$U = \sum_i d'_i/T_i$, SF-BF places all stations. Conversely (cf.
Proposition~\ref{prop:ce}), a station of level $\ge 1$ must satisfy
$d' \le t_0(1 - U_0)$ where $U_0$ is the lower-level utilization: the
threshold is tight up to the requesting station's own density share.
\end{lemma}

\begin{proof}[Proof sketch]
Suppose a request of size $q$ fails at level $j$. Every free block is then
smaller than $q$, and since each of the sizes $q/2, q/4, \dots$ carries at
most $2^j$ blocks (Lemma~\ref{lem:counts}), the total free measure is
below $2^j q$. On the other hand, the stations placed so far occupy at
most a fraction $U$ of the level-$j$ window of length $t_0 2^j$, leaving
free measure at least $t_0 2^j (1-U) \ge 2^j q$ under the stated
granularity -- a contradiction. Full proofs are
in \ifdefined\TECHREPORT the appendix\else the technical report\fi, and
all lemmas are additionally machine-checked on $\sim$$3{,}000$ randomized
harmonic instances.
\end{proof}

\begin{theorem}
\label{thm:approx}
Under the granularity assumption $d_i \le t_0/4$, \textsc{Harmonic-Greedy}
(analysis variant, $\rho^\star = 1/4$) returns a feasible schedule of cost
at most $8 \cdot \mathrm{OPT}$.
\end{theorem}

\begin{corollary}[candidate-argmin anchor]
\label{cor:anchor}
The rounded cost, as a function of the anchor phase
$\theta \in [0,1)$ ($t_0 = T_{\min} 2^{-\theta}$), is piecewise strictly
decreasing with minima exactly at the folded candidates; hence the
candidate-argmin anchor satisfies
$\min_{\mathrm{cand}} C \le \mathbb{E}_\theta[C] = (1/\ln 2)\sum_i a_i T_i$,
improving Theorem~\ref{thm:approx} deterministically to
$4/\ln 2 \approx 5.77$.
\end{corollary}

\begin{theorem}[unconditional]
\label{thm:uncond}
Without any granularity assumption, scaling the relaxation solution by
$s = 6$ (which simultaneously enforces the density budget and, via
Lemma~\ref{lem:longsp}, the long-SP margin $t_0 > 3\dmax$) yields a
$12$-approximation, and $6/\ln 2 \approx 8.66$ with the candidate-argmin
anchor. The constant $6$ is optimal for this accounting.
\end{theorem}

\begin{example}[concluded]
On Example~\ref{ex:run} after rounding, $U = 4/8 + 4/16 + 4/32 = 0.875$, so
Lemma~\ref{lem:packing} would require $d \le t_0(1-U) = 1$\,ms; the instance
violates the threshold ($d = 4$\,ms) yet packs anyway -- the condition is
sufficient, not necessary, and the certified variant's $\rho^\star = 1/4$
exists precisely to guarantee it in the worst case. The achieved cost,
$1.16\times$ the relaxation bound, is typical of the gap between certificate
and behavior: across every regime of Table~\ref{tab:regimes} the practical
variant lands within $1.05$--$1.25\times$ the relaxation lower bound --
itself unachievable in general -- against the $5.77\times$ worst-case
certificate. The bound's role is to exclude pathological instances, not
to predict typical loss.
\end{example}

\begin{proposition}[peak AoI]
\label{prop:peak}
All of the above holds verbatim for weighted peak AoI: by
Proposition~\ref{prop:renewal} the objective is again linear in $T$ with
coefficients $a^{\mathrm{peak}}_i = w_i/p_i$.
\end{proposition}

For $N = 2$ the problem admits an \emph{exact} solution, with a
structural twist:

\begin{theorem}[exact $N{=}2$]
\label{thm:n2}
Feasibility forces a rational period ratio $T_2/T_1 = p/q$ with
$T_1 \ge q(d_1 + d_2)$. For each $q$ the cost
$J(r,q) = (a_1 + a_2 r)\max(\tmin_1, \tmin_2/r, q(d_1{+}d_2))$ is quasiconvex
in $r = p/q$ -- strictly decreasing while the $\tmin_2/r$ term is active,
then strictly increasing. Its integer-numerator optimum is therefore one
of the two grid neighbors of the continuous minimizer, and $q$ ranges over
a finite, explicit set, yielding an exact $O(Q)$ algorithm (verified
against exhaustive search on $4{,}000$ instances). Notably, \emph{non-integer} ratios
are strictly optimal on an open set of instances -- both energy floors can
bind only at a fractional ratio -- so the harmonic restriction has a real
price already at $N = 2$, precisely the rounding loss our analysis charges
for. Full proof: \ifdefined\TECHREPORT appendix\else technical report\fi.
\end{theorem}

\ifdefined\TECHREPORT\else
Full proofs of all results appear in the appendix of the technical-report
version.
\fi

\section{Implementation in ns-3}
\label{sec:implementation}

We implement the full system in ns-3.\footnote{Code, data, and a
reproduction guide are available at
\url{https://github.com/Haoyu2/age-optimal-twt} -- an ns-3 \texttt{contrib}
module plus the \texttt{src/wifi} mechanism patch, with the raw sweep results,
released to support artifact evaluation.\ifdefined\TECHREPORT\else{} The
extended version with full proofs is available as arXiv:2608.21596.\fi} It comprises a
\TwtPSM{}
realizing individual implicit TWT as a subclass of the 802.11 power-save
architecture, and the scheduler as a pluggable policy module. Three
mechanism elements proved essential, each corresponding to a measured
failure mode:
(i) \emph{deferral, not wake-on-traffic}: channel access requested outside
an SP must not wake the radio -- queued frames contend automatically at the
next SP start;
(ii) \emph{SP-boundary queue gating}: without it, pending retransmissions
keep the station awake until the fade lifts (we measured a retry storm of
$1.5\times10^5$ transmissions), erasing the energy savings. We gate by
\emph{freezing} rather than dropping: at SP end we mark the station's
per-access-category EDCA queues blocked through ns-3's existing power-save
block-reason mechanism (the same hook the standard PS manager uses), and
release it at the next SP start. Blocked frames are not dequeued for
transmission, so their sequence numbers, retry counts, and Block-Ack state
are preserved untouched; the deferred updates simply ride the next SP. This
restores the exact $k$-attempts-per-SP structure the model assumes and keeps
the PHY in sleep between SPs;
(iii) \emph{TSF-anchored offsets}: per-station offsets must live on a common
absolute time grid -- otherwise the computed disjointness has no shared
reference and is silently lost.
We further document protocol findings on TWT deployment: beacon-loss
monitoring disassociates dozing stations; downlink frames (including
ADDBA responses) buffered for a dozing station deadlock the uplink through
asymmetric Block Ack state; and doze announcements are best-effort by
nature -- a PM indication cannot be acknowledged over a channel the AP
cannot hear, which is precisely why TWT is schedule-based.

\textbf{Scheduler cost.} The scheduler runs on the AP once per
(re)association or renegotiation epoch, not per frame. Lines~1--2 of
Algorithm~\ref{alg:hg} are $O(N)$ (the bisection runs a fixed number of
iterations to machine precision), anchor selection evaluates $N$ candidate
grids at $O(N)$ each, and packing touches each SP instance of the
hyperperiod. Measured on a single commodity-VM core (median of 20 runs,
heterogeneous synthetic workloads): $0.7$\,ms at $N = 64$ stations,
$9$\,ms at $N = 256$, and $43$\,ms at $N = 512$ -- negligible against
association signaling. The
runtime state is one triple per station plus the per-AC gating flags, and
the data path adds only the two $O(1)$ queue gate/ungate operations per SP
boundary, reusing the existing power-save block-reason hook.

\section{Evaluation}
\label{sec:evaluation}

\textbf{Setup.} All experiments use 802.11ax at 2.4\,GHz over a 20\,MHz
channel (HeMcs3 unless stated otherwise), with 10--50 stations of
heterogeneous weights, duty budgets, and per-SP fading. Each station
delivers one status update per wake interval, and the AP integrates AoI
exactly, event by event. Runs last $300$\,s over 10 seeds; we report
$95\%$ confidence intervals.

\textbf{Baselines.} We compare \textsc{Harmonic-Greedy} against two
baselines. \emph{Equal-interval} gives all stations a single common period
(the naive TWT schedule). \emph{Energy-greedy} is a strong heuristic in the
spirit of energy-efficiency-first TWT schedulers such as
TASPER~\cite{tasper25}: it gives each station its shortest budget-feasible
period $\max(\tmin_i, \dmax)$, uniformly stretches to feasibility if
overbooked, and reuses our rounding, packing and safeguard. It is
period-\emph{diverse} but AoI-weight-\emph{unaware}, so it isolates the value
of our weighted period assignment: any margin of \textsc{Harmonic-Greedy}
over Energy-greedy is attributable to that step alone. The two baselines let
us separate two questions -- does period diversity help (Energy-greedy
vs.\ Equal-interval), and does AoI-weighting help beyond it
(\textsc{Harmonic-Greedy} vs.\ Energy-greedy). Slot-based AoI policies
(max-weight, Whittle index~\cite{kadota18,tripathi24}) are not directly
comparable: they assume the scheduler can reach any station in any slot,
whereas a dozing TWT station is unreachable between its negotiated SPs, and
re-deriving index policies over wake agreements is the stochastic-arrivals
extension of Section~\ref{sec:conclusion}. The relaxation lower bound plays
the comparison role instead: it bounds the cost of \emph{every} feasible
TWT schedule -- hence of any AoI policy adapted to TWT -- and
\textsc{Harmonic-Greedy} stays within $5$--$25\%$ of it below.

\textbf{Energy--freshness frontier.} Figure~\ref{fig:pareto} sweeps the
duty-cycle budgets ($8\times$ weight skew): the schedulers trace the
AoI--energy trade-off, with \textsc{Harmonic-Greedy} dominating in the
energy-constrained region where ideal-period spread exceeds an octave and
coinciding with the baseline where the safeguard detects collapse.

\begin{figure}[t]
\centering
\IfFileExists{figures/pareto.pdf}{%
\includegraphics[width=0.95\linewidth]{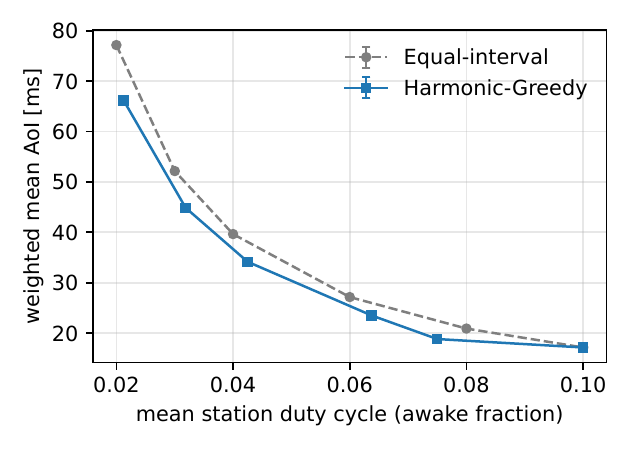}}{%
\fbox{\parbox{0.9\linewidth}{\centering\vspace{2em}
Figure placeholder: AoI--energy Pareto frontier
(\texttt{scripts/make\_pareto.py})\vspace{2em}}}}
\caption{Weighted mean AoI vs.\ mean duty cycle across energy budgets
($8\times$ weight skew, 10 stations, 5 seeds).}
\label{fig:pareto}
\end{figure}

\textbf{Model validation.} (Fig.~\ref{fig:validation},
Table~\ref{tab:validation}) Across the loss grid
$e \in \{0, 0.1, \dots, 0.6\}$, measured AoI tracks
$\delta + T(1/p - 1/2)$ with mean error ${\sim}1\%$: $0.003\%$ at $p = 1$
on every seed, and a worst single seed of $4.4\%$ at $p{=}0.4$ --
consistent with the geometric-tail variance of a $300$\,s run. The
measured renewal components $\mathbb{E}[G]$ and
$\mathbb{E}[G^2]/2\mathbb{E}[G]$ are within $1\%$ of their predicted
values $T/p$ and $T(1/p - 1/2)$.

\begin{figure}[t]
\centering
\IfFileExists{figures/validation.pdf}{%
\includegraphics[width=0.92\linewidth]{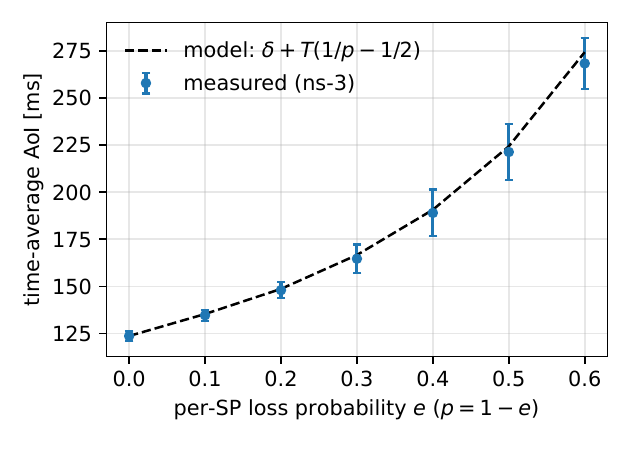}}{}
\caption{Single-station validation: measured time-average AoI (3 seeds,
$95\%$ CI) vs the renewal model, across per-SP loss probabilities
($T{=}100$\,ms, $d{=}8$\,ms).}
\label{fig:validation}
\end{figure}

\begin{figure}[t]
\centering
\IfFileExists{figures/regimes.pdf}{%
\includegraphics[width=0.92\linewidth]{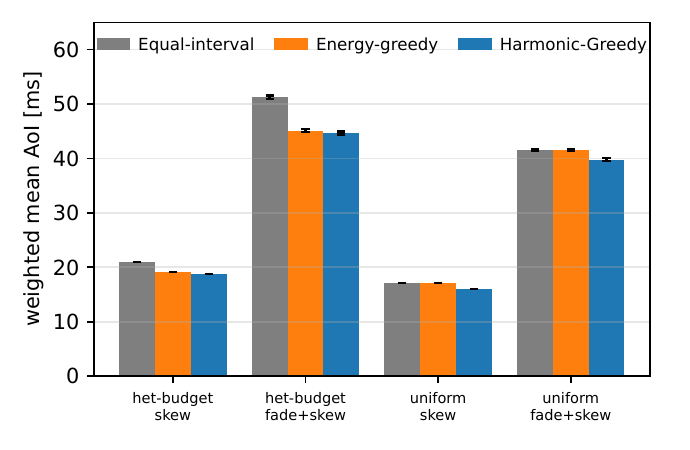}}{}
\caption{Weighted mean AoI by regime (10 seeds, $95\%$ CI). With
heterogeneous budgets (left) the strong energy-greedy baseline already
matches \textsc{Harmonic-Greedy}; with uniform budgets and binding density
(right) it collapses to equal-interval and only AoI-weighting wins.}
\label{fig:regimes}
\end{figure}

\textbf{Period diversity vs.\ AoI-weighting.} (Fig.~\ref{fig:regimes},
Table~\ref{tab:regimes}) The two effects separate cleanly by regime.
\emph{When budgets are heterogeneous}, the per-station energy floors
$\tmin_i = d_i/\rho_i$ already pin the periods, so ``minimize every period
then pack'' is near-optimal: Energy-greedy captures essentially the whole
gain over equal-interval and \textsc{Harmonic-Greedy} only matches it
(het-budget skew: $\ResSkewHarmonic$ vs $\ResSkewEnergy$\,ms, both vs
$\ResSkewEqual$ for equal-interval; fading+skew similar). We report this
plainly -- period diversity, not AoI-weighting, drives those gains, and the
LB column confirms both schedulers already sit near the optimum there.
\emph{When budgets are uniform}, Energy-greedy collapses to equal-interval
(identical floors give identical periods), and only the AoI-weighted
assignment redistributes the binding schedule density toward the stations
that matter: under uniform budgets with binding density,
\textsc{Harmonic-Greedy} wins $\ResUbindGainE\%$ with weight skew
($\ResUbindHarmonic$ vs $\ResUbindEnergy$\,ms) and $\ResUbindfadeGainE\%$
with skew plus fading -- the latter via the channel-aware coefficient
$a_i \propto 1/p_i$. This is the regime that isolates the paper's
contribution. Weighted \emph{peak} AoI tracks the same pattern
(Proposition~\ref{prop:peak}); duty cycles equal $d_i/T_i$ to four decimals
and SPs never overlap in all runs.

\textbf{Near-optimality.} The LB column of Table~\ref{tab:regimes} is the
relaxation optimum -- a lower bound on \emph{any} feasible schedule.
\textsc{Harmonic-Greedy} is the only scheduler that stays close to LB in
\emph{every} regime; Energy-greedy is near-LB when energy floors pin the
periods but drifts away once the binding density must be redistributed by
weight or channel rather than by budget. The residual gap to LB is the price
of harmonic rounding and contiguous (non-preemptive) packing, which
Proposition~\ref{prop:ce} shows is unavoidable in general.

\textbf{Scaling and decomposition.} (Fig.~\ref{fig:scaling},
Fig.~\ref{fig:mixedmcs}) The pipeline runs cleanly to $N = 50$ stations
(exact duty cycles, zero SP overlap, predictions within $2\%$).
Heterogeneous SP sizes -- alternating HeMcs0/HeMcs7, $700$-byte updates,
airtimes differing ${\sim}3\times$ -- give the cleanest decomposition of the
two effects at scale (3 seeds, CIs within $0.1$\,ms). At $N{=}50$:
equal-interval $38.5$\,ms; Energy-greedy $29.5$\,ms (a $23\%$ gain from
period diversity alone, as fast stations need not wait for slow ones); and
\textsc{Harmonic-Greedy} $24.6$\,ms (a further $17\%$ from AoI-weighting, for
$36\%$ over equal-interval). Both ingredients contribute and
\textsc{Harmonic-Greedy} beats the strong baseline outright here, because
diverse airtimes create both period diversity \emph{and} a binding density
that AoI-weighting can exploit. ($N{=}20$ shows the same split:
$16.8 \to 13.3 \to 11.2$\,ms.)

\begin{figure}[t]
\centering
\IfFileExists{figures/mixedmcs.pdf}{%
\includegraphics[width=0.92\linewidth]{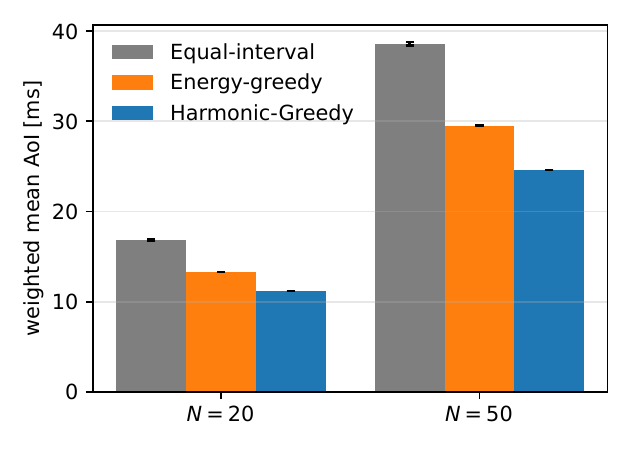}}{}
\caption{Heterogeneous modulation rates (HeMcs0/HeMcs7 alternating,
$700$-byte updates, $8\times$ weight skew, 3 seeds): the gain decomposes into
period diversity (energy-greedy over equal-interval) and AoI-weighting
(\textsc{Harmonic-Greedy} over energy-greedy); $36\%$ total at $N{=}50$.}
\label{fig:mixedmcs}
\end{figure}

\textbf{Certified vs.\ practical variant.} \textsc{Harmonic-Greedy} has two
configurations that share the relaxation, anchor rule, and safeguard but
differ in the packing stage (Sec.~\ref{sec:algorithm}). The
\emph{certified} variant (dyadic reservations, $\rho^\star{=}1/4$) is the
one Theorems~\ref{thm:approx}--\ref{thm:uncond} bound; the \emph{practical}
variant ($\rho^\star{=}0.9$, first-fit of true SP durations) is the one that
produces the gains above. The conservative $\rho^\star{=}1/4$ is forced by
the two factor-of-two losses the worst-case packing analysis must absorb in
series: rounding each SP duration up to a dyadic reservation can double the
schedule density, and the packing threshold (Lemma~\ref{lem:packing}) admits
a request only when $d' \le t_0(1-U)$, which the count invariant satisfies
for $U \le 1/2$; starting from $\rho^\star{=}1/4$ leaves exactly the headroom
for both. The practical variant instead packs the true SP durations
greedily, so neither loss applies and density $0.9$ is feasible -- but it has
no worst-case certificate. The two are thus not interchangeable: the theorem certifies the
\emph{template} (relaxation, anchor rule, packing, safeguard), while the
high-density practical variant is what we evaluate throughout. Under its
conservative budget the certified variant's safeguard returns the uniform
schedule in the heterogeneous-budget regimes, so it tracks equal-interval
there rather than the practical variant; closing this gap -- a tighter
packing analysis that certifies a density target near $0.9$ -- is open
(Sec.~\ref{sec:conclusion}).

\begin{figure}[t]
\centering
\IfFileExists{figures/scaling.pdf}{%
\includegraphics[width=0.92\linewidth]{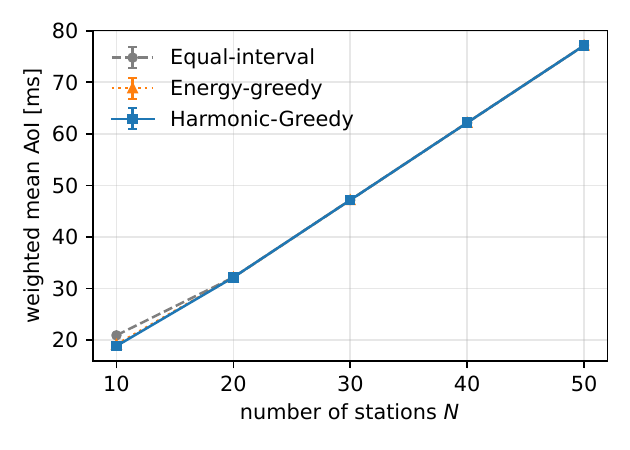}}{}
\caption{Scaling with the number of stations ($8\times$ weight skew,
heterogeneous budgets, 3 seeds): all three schedulers track closely once the
schedule is capacity-bound (energy floors pin the periods).}
\label{fig:scaling}
\end{figure}

\begin{table}[t]
\centering
\caption{Single-station model validation (summary).}
\label{tab:validation}
\begin{tabular}{lcc}
\toprule
Regime & Predicted & Error \\
\midrule
$p\approx 1$, $T{=}100/8$\,ms & $\delta + T/2$ & $0.016\%$ \\
$p\approx 1$, $T{=}50/4$\,ms & $\delta + T/2$ & $0.065\%$ \\
block fading $p{=}0.8/0.6/0.4$ & $\delta + T(1/p - 1/2)$ &
  $\le 1.2\%$ (mean) \\
\bottomrule
\end{tabular}
\end{table}

\begin{table}[t]
\centering
\caption{Weighted mean AoI (ms; mean $\pm$ 95\% CI over 10 seeds),
10 stations, 300\,s. Equal = equal-interval, EG = energy-greedy (strong
baseline), \textsc{H-G} = ours, LB = relaxation lower bound (unachievable).
Top: heterogeneous budgets -- EG already near-optimal, \textsc{H-G} matches.
Bottom: uniform budgets, binding density -- EG collapses to Equal, only
\textsc{H-G} (AoI-weighting) wins.}
\label{tab:regimes}
\begin{tabular}{lcccc}
\toprule
Regime & Equal & EG & \textsc{H-G} & LB \\
\midrule
het.\ budget, skew & \ResSkewEqual & \ResSkewEnergy & \ResSkewHarmonic &
  \ResSkewLB \\
het.\ budget, fade+skew & \ResFadingSkewEqual & \ResFadingSkewEnergy &
  \ResFadingSkewHarmonic & \ResFadingSkewLB \\
\midrule
unif.\ budget, skew & \ResUbindEqual & \ResUbindEnergy &
  \textbf{\ResUbindHarmonic} & \ResUbindLB \\
unif.\ budget, fade+skew & \ResUbindfadeEqual & \ResUbindfadeEnergy &
  \textbf{\ResUbindfadeHarmonic} & \ResUbindfadeLB \\
\bottomrule
\end{tabular}
\end{table}

\section{Conclusion}
\label{sec:conclusion}

TWT turns freshness-versus-energy into a schedule design problem. We
formulated it with AoI as the objective, derived a renewal-exact model
validated at the packet level, exposed the non-preemptive structure that
separates TWT packing from classical utilization bounds, and built
\textsc{Harmonic-Greedy} -- to our knowledge the first TWT scheduler with
approximation guarantees -- together with an exact $N{=}2$ solution and an
open ns-3 implementation.

For practitioners, the evaluation reads as a regime map. When per-station
energy floors pin the periods, energy-first scheduling is already
near-optimal and AoI-weighting adds nothing. When the schedule density
binds under uniform budgets or mixed
modulation rates, AoI-weighted assignment is the only policy that
redistributes air time toward the stations that matter, worth
$4$--$36\%$.

Four directions remain open. \emph{Tighter packing certificates}: the
certified $\rho^\star{=}1/4$ is conservative only because the worst-case
analysis absorbs two compounding factor-of-two losses in series -- dyadic
duration rounding can double the density, and the non-preemptive threshold
(a consequence of the density-1 infeasibility of
Proposition~\ref{prop:ce}) admits requests only up to half the window; the
practical $\rho^\star{=}0.9$ variant never triggers either worst case, so
an amortized packing argument would likely certify a density target near
$0.9$ -- the single highest-leverage next step. \emph{Stochastic arrivals}:
index policies over wake agreements, with asymptotic-optimality machinery.
\emph{Exact solutions beyond $N{=}2$}. \emph{Downlink AoI}: broadcast TWT
with AP-side schedule-aware delivery.

\section*{Acknowledgment}
This material is based upon work supported in part by the National
Science Foundation under Award No.\ CNS-2450832. Any opinions, findings,
and conclusions or recommendations expressed in this material are those
of the authors and do not necessarily reflect the views of the National
Science Foundation.

\bibliographystyle{IEEEtran}
\ifdefined\TECHREPORT\else\balance\fi
\bibliography{refs}

\ifdefined\TECHREPORT
% Full proofs appendix -- included when \TECHREPORT is defined (make techreport)
\appendix

\section{Proofs for the packing analysis}

Fix a unit $u > 0$ with $t_0 = 2^L u$. A \emph{block of size} $s = 2^k u$ is
an interval $[rs, (r{+}1)s)$, $r \in \mathbb{N}$ (aligned). Periods are
$T_i = t_0 2^{j_i}$ (level $j_i$, window $W_j = t_0 2^j$); reservations are
$d'_i = 2^{k_i} u \le t_0$; $U = \sum_i d'_i / T_i$. SF-BF processes levels
in increasing order, duplicating the free set at each level transition
(valid: everything placed so far is $W_j$-periodic), and serves each
request from a minimum-size sufficient free block, carving from the left.

\begin{lemma}[staircase carving]
\label{app:staircase}
Carving a request $q$ from the left end of a block $B = [b, b+s)$,
$q \le s$, leaves $[b+q, b+s)$, which is the disjoint union of exactly one
aligned block of each size $q, 2q, \dots, s/2$ (empty if $q = s$).
\end{lemma}

\begin{proof}
With $s = 2^M q$, define $I_l = [b + 2^l q,\, b + 2^{l+1} q)$ for
$l = 0, \dots, M{-}1$. These are disjoint with union $[b+q, b+s)$, and
$I_l$ has size $2^l q$ and left endpoint $b + 2^l q$, a multiple of
$2^l q$ since $b$ is a multiple of $s \ge 2^{l+1} q$.
\end{proof}

\begin{lemma}[count invariant]
\label{app:counts}
(i) A successful SF-BF allocation creates new free blocks only at sizes
whose current count is $0$. (ii) At the start of level-$j$ processing,
every size has count at most $2^j$. Both bounds are tight.
\end{lemma}

\begin{proof}
(i) Let $q$ be the request and $b$ the minimum size of a free block with
$b \ge q$. The staircase sizes of Lemma~\ref{app:staircase} lie in
$[q, b)$; by minimality of $b$ no free block has size in $[q, b)$, so each
created size moves from count $0$ to $1$. (ii) Induction: level $0$ starts
from the single block $[0, t_0)$; by (i) the maximum count never increases
during a level; the transition duplicates the free set, doubling counts.
Tightness of (ii): when intermediate levels carry no requests, the
transition doubling compounds unimpeded, so a level-$j$ request can meet
$2^j$ blocks of one size (construction verified mechanically).
\end{proof}

\begin{lemma}[failure bound]
\label{app:failure}
If a level-$j$ request $q$ fails, the free measure satisfies
$|F| < 2^j q$.
\end{lemma}

\begin{proof}
All free blocks have size $\le q/2$; distinct sizes are
$q/2, q/4, \dots, u$, each with count $\le 2^j$
(Lemma~\ref{app:counts}), so $|F| \le 2^j (q - u) < 2^j q$.
\end{proof}

\begin{proof}[Proof of the packing threshold (Lemma~\ref{lem:packing} in the paper)]
Suppose a request $q = d'_i$ at level $j$ fails. The occupied measure
within $[0, W_j)$ is $W_j$ times the density of the stations placed so
far, at most $W_j U$, so $|F| \ge W_j (1 - U) = 2^j t_0 (1 - U) \ge 2^j q$
by hypothesis, contradicting Lemma~\ref{app:failure}. For the converse,
generalize Proposition~\ref{prop:ce}: let level-$0$ stations occupy measure
$(1-\epsilon) t_0$ of the base window with one contiguous free gap of
length $\epsilon t_0$. A station of level $\ge 1$ must place its
contiguous SP inside one periodic copy of the gap (it cannot straddle the
level-$0$ pattern), so $d' \le \epsilon t_0 = t_0 (1 - U_0)$ is necessary,
where $U_0$ is the lower-level utilization.
\end{proof}

\section{Proofs for the approximation chain}

\begin{proof}[Proof of Theorem~\ref{thm:approx} (8-approximation)]
\emph{Lower bound:} any feasible schedule has disjoint SPs, hence density
$\le 1$ and $T_i \ge \tmin_i$, so $\mathrm{OPT} - C \ge J^*$ where $J^*$
solves \eqref{eq:relax} (with the $\dmax$ constraint dropped; it is also
valid by Lemma~\ref{lem:longsp}). \emph{Density scaling:} for $\rho \le 1$,
$T^*(1)/\rho$ is feasible for density target $\rho$ (density scales by
$\rho$, the floors still hold), so $J(T(\rho)) \le J^*/\rho$; we run the
relaxation at $\rho^\star = 1/4$. \emph{Rounding:} round each $T_i$ up to
the grid $t_0 2^j$, $t_0 \le T_{\min}$ (factor $< 2$ on cost; density only
decreases); round each $d_i$ up to a dyadic divisor $d'_i$ of $t_0$
(factor $< 2$, possible since $d_i \le t_0/4 < t_0$). The resulting
density satisfies $U \le 2 \cdot \tfrac14 = \tfrac12$, and
$d'_i \le t_0/2 = t_0(1 - \tfrac12) \le t_0 (1 - U)$. \emph{Packing:}
Lemma~\ref{lem:packing} places everything; each station transmits only $d_i$ inside its
reservation, so true duty cycles satisfy
$d_i / T'_i \le d_i / \tmin_i = \rho_i$ -- reservations cost no energy.
\emph{Chain:} cost $\le 2 \cdot J(T(1/4)) \le 2 \cdot 4 J^* \le
8\,(\mathrm{OPT} - C)$.
\end{proof}

\begin{proof}[Proof of Corollary~\ref{cor:anchor} (candidate-argmin anchor)]
For $\theta \in [0,1)$ let $t_0(\theta) = T_{\min} 2^{-\theta}$ and
$\ell_i(\theta) = 2^{\lceil z \rceil - z}$, $z = y_i + \theta$,
$y_i = \log_2(T_i/T_{\min})$, the per-station rounding factor.
(a) For $\theta$ uniform, $\mathrm{frac}(z)$ is uniform and
$\mathbb{E}[2^{\lceil z\rceil - z}] = \int_0^1 2^{1-f} df = 1/\ln 2$;
by linearity $\mathbb{E}_\theta[C(\theta)] = (1/\ln 2)\sum_i a_i T_i$.
(b) On any interval free of breakpoints
($\theta$ with $\mathrm{frac}(y_i + \theta) = 0$ for some $i$), every
$\mathrm{frac}(y_i+\theta)$ increases without wrapping, so every
$\ell_i$ -- and hence $C$ -- strictly decreases. As $\theta$ approaches a
breakpoint from below, the wrapping station's factor tends to $1$, its
value \emph{at} the breakpoint (an exact power of two rounds to itself),
so $C$ is left-continuous there and attains its infimum at a breakpoint.
The breakpoints are $\theta_c = 1 - \mathrm{frac}(y_c)$, i.e.
$t_0(\theta_c) = T_c/2^{\lceil y_c \rceil}$: exactly the folded
candidates. (c) $\min_{\mathrm{cand}} C = \min_\theta C \le
\mathbb{E}_\theta[C]$.
\end{proof}

\begin{proof}[Proof of Theorem~\ref{thm:uncond} (unconditional)]
By Lemma~\ref{lem:longsp} the relaxation may include $T_i \ge \dmax$. Set
$\hat T = 6\,T^{*\prime}$: density $\le 1/6$, every period $\ge 6\dmax$,
cost $6 J^{*\prime}$. Any folded-candidate anchor satisfies
$t_0 > \hat T_{\min}/2 \ge 3 \dmax$. Round periods (factor $<2$, or
$1/\ln 2$ in the candidate-argmin sense) and reservations
($d_i \le \dmax < t_0/3 < t_0$, factor $< 2$): the density becomes
$U' \le 2/6 = 1/3$ and $d'_i < 2\dmax < \tfrac23 t_0 \le t_0 (1 - U')$,
so Lemma~\ref{lem:packing} packs. Cost $\le 2 \cdot 6\, J^{*\prime} \le 12\,\mathrm{OPT}$.
Optimality of $s = 6$ for this accounting: the packing condition reads
$2\dmax \le \tfrac{s}{2}\dmax\,(1 - \tfrac{2}{s})$, i.e. $s \ge 6$.
\end{proof}

\section{Exact solution for $N = 2$}

Write $J = a_1 T_1 + a_2 T_2$, $T_1 \le T_2$, energy floors
$\tmin_1, \tmin_2$.

\begin{lemma}
\label{app:rational}
Any feasible two-station schedule has $T_2/T_1$ rational.
\end{lemma}

\begin{proof}
If $T_2/T_1$ is irrational, the differences
$\{ j T_2 - k T_1 : j, k \in \mathbb{N} \}$ are dense in $\mathbb{R}$ (an
irrational rotation equidistributes), so some pair of SP start times
differs by a value in $(-d_1, d_2)$ -- and those two SPs overlap.
\end{proof}

\begin{lemma}
\label{app:phase}
Let $T_2/T_1 = p/q$ reduced, $p \ge q \ge 1$. The schedule is feasible iff
$T_1 \ge q (d_1 + d_2)$.
\end{lemma}

\begin{proof}
Since $\gcd(p,q) = 1$, the starts of station~2's SPs modulo $T_1$ form $q$
equally spaced phases (spacing $T_1/q$). Station~1's contiguous SP cannot
straddle an SP of station~2, so it must fit in one inter-phase gap of
length $T_1/q - d_2$: feasibility $\iff T_1/q - d_2 \ge d_1$. Sufficiency:
place $\phi_1$ in any gap.
\end{proof}

\begin{theorem}[exact $N{=}2$]
\label{app:n2}
Let $J(r, q) = (a_1 + a_2 r)\max(\tmin_1,\, \tmin_2/r,\, q(d_1{+}d_2))$ and
\begin{equation*}
r_q^* = \max\!\big(1,\ \tfrac{\tmin_2}{\max(\tmin_1,\, q(d_1{+}d_2))}\big),
\quad
Q = \big\lceil \tfrac{J(\lceil r_1^* \rceil, 1)}{(a_1{+}a_2)(d_1{+}d_2)} \big\rceil.
\end{equation*}
Then the optimum equals
$\min_{q \le Q} \min_{p} J(p/q, q)$ over
$p \in \{\lfloor q r_q^* \rfloor, \lceil q r_q^* \rceil\}$ with
$p \ge q$ -- an $O(Q)$ exact algorithm.
\end{theorem}

\begin{proof}
Lemmas~\ref{app:rational}--\ref{app:phase} reduce the problem to
$\min J(p/q, q)$ over reduced rationals with the stated feasible region:
for fixed $r$ and $q$ the objective increases in $T_1$, so the optimal
$T_1$ is the smallest feasible value -- the active maximum inside $J$. For
fixed $q$, $J(\cdot, q)$ is quasiconvex on $[1, \infty)$: while
$\tmin_2/r$ is the active maximum, $J = a_1 \tmin_2/r + a_2 \tmin_2$
strictly decreases; once the maximum is constant in $r$, $J$ strictly
increases. Hence the integer numerator minimum is attained at one of the
two grid neighbors of the continuous minimizer $r_q^*$. Finally
$J(r, q) \ge (a_1 + a_2)\, q (d_1 + d_2)$, so all
$q > Q$ are dominated by the $q = 1$ candidate defining $Q$.
(Non-reduced fractions replicate a smaller-$q$ ratio with a tighter
packing floor and are dominated.) Non-integer optimality:
$\tmin_1 = 10$, $\tmin_2 = 15$, $d_1 + d_2 = 1$ gives ratio $3/2$ cost
$10 a_1 + 15 a_2$, beating every integer ratio; continuity extends the
strict gap to a neighborhood.
\end{proof}

\fi

\end{document}